\documentclass[11pt]{article}
\usepackage{fullpage}
\usepackage{algorithmic}
\usepackage{epsfig}
\usepackage{epstopdf}
\usepackage{graphicx}
\usepackage{latexsym}
\usepackage{amsmath}
\usepackage{amsfonts}
\usepackage{amssymb}

\usepackage{mathrsfs}
\usepackage{float}
\usepackage{pifont}
\usepackage{yhmath}
\usepackage{bbm}
\usepackage{eufrak}
\usepackage{amsthm}
\usepackage{booktabs}
\usepackage[usenames,dvipsnames]{color}
\usepackage{stmaryrd}
\usepackage{wasysym}
\usepackage{array}
\usepackage[ruled]{algorithm2e}

\usepackage{bigstrut,multirow,rotating}

\newtheorem{theorem}{Theorem}[section]
\newtheorem{lemma}[theorem]{Lemma}

\newtheorem{claim}[theorem]{Claim}

\newtheorem{proposition}[theorem]{Proposition}

\theoremstyle{remark}

\makeatletter
\newcommand\figcaption{\def\@captype{figure}\caption}
\newcommand\tabcaption{\def\@captype{table}\caption}
\makeatother

\begin{document}

\title{\bf On Maximizing a Monotone $k$-Submodular Function under a Knapsack Constraint 
}

\date{}
\maketitle

\vspace{-3em}
\begin{center}

\author{Zhongzheng Tang$^1$\quad Chenhao Wang$^{2,3}$\quad Hau Chan$^4$\\
${}$\\
1  Beijing University of Posts and Telecommunications\\
2 Beijing Normal University-Zhuhai\\
3  BNU-HKBU United International College\\
4 University of Nebraska-Lincoln\\
\medskip
tangzhongzheng@amss.ac.cn, ~chenhwang@bnu.edu.cn,~
hchan3@unl.edu}
\end{center}
\vspace{1em}


\maketitle

\begin{abstract}
We study the problem of maximizing a non-negative monotone $k$-submodular function $f$ under a knapsack constraint, where a $k$-submodular function is a natural generalization of a submodular function to $k$ dimensions. We present a deterministic $(\frac12-\frac{1}{2e})\approx 0.316$-approximation algorithm that evaluates $f$ $O(n^4k^3)$ times, based on the result of Sviridenko (2004) on submodular knapsack maximization.\footnote{This manuscript is published in Operations Research Letters, but there is an error
in the proof of Theorem 1. We provide a corrigendum in the end of this manuscript. }
\end{abstract}

\section{Introduction}\label{intro}
\vspace{-2mm}A $k$-submodular function is a generalization of submodular function, where the input consists of $k$ disjoint subsets of the domain, instead of a single subset.   The $k$-submodular maximization problem has been studied in the unconstrained setting \cite{ward2016maximizing}, under cardinality constraints \cite{ohsaka2015monotone}, and under matroid constraints \cite{sakaue2017maximizing}, because it appears in a broad range of applications (e.g., influence maximization with $k$ kinds of topics, and sensor placement with $k$ kinds of sensors \cite{ohsaka2015monotone}). 


Let $V$ be a finite set. 
Let $(k+1)^V:=\{(X_1,\ldots,X_k)~|~X_i\subseteq V ~\forall i\in[k], X_i\cap X_j=\varnothing ~\forall i\neq j\}$ be the family of $k$ disjoint sets, where $[k]:=\{1,\ldots,k\}$. A function $f:(k+1)^V\rightarrow \mathbb R$ is called \emph{$k$-submodular} \cite{huber2012towards}, if for any $\mathbf x=(X_1,\ldots,X_k)$ and $\mathbf y=(Y_1,\ldots,Y_k)$ in $(k+1)^V$, we have
$$f(\mathbf x)+f(\mathbf y)\ge f(\mathbf x \sqcup\mathbf y)+f(\mathbf x \sqcap\mathbf y),$$
where
$$\mathbf x\sqcup\mathbf y:=\left(X_1\cup Y_1\backslash(\bigcup_{i\neq 1}X_i\cup Y_i),\ldots,X_k\cup Y_k\backslash(\bigcup_{i\neq k}X_i\cup Y_i)\right),$$
$$\mathbf x\sqcap\mathbf y:=\left(X_1\cap Y_1,\ldots,X_k\cap Y_k\right).$$

Denote $\mathbf x\preceq\mathbf y$, if $\mathbf x=(X_1,\ldots,X_k)$ and $\mathbf y=(Y_1,\ldots,Y_k)$ with $X_i\subseteq Y_i$ for each $i\in[k]$. 
Define the marginal gain when adding item $a$ to the $i$-th dimension of $\mathbf x$ to be
$$\Delta_{a,i}(\mathbf x):=f(X_1,\ldots,X_{i-1},X_i\cup\{a\},X_{i+1},\ldots,X_k)-f(\mathbf x).$$
It is not hard to see that, a $k$-submodular function $f$ satisfies the \emph{orthant submodularity}, that is, $$\Delta_{a,i}f(\mathbf x)\ge \Delta_{a,i}f(\mathbf y),~\text{~for any~}~ \mathbf x,\mathbf y\in(k+1)^V~\text{~with~}~ \mathbf x\preceq\mathbf y, a\notin\cup_{j\in[k]}Y_j, i\in[k].$$
A function $f:(k + 1)^V\rightarrow\mathbb R$ is called \emph{monotone}, if $f(\mathbf x)\le f(\mathbf y)$ for any $\mathbf x\preceq\mathbf y$. Ward and {\v{Z}}ivn{\`y} \cite{ward2016maximizing} shows that when monotonicity holds, $f$ is $k$-submodular if and only if it is orthant submodular.

In this note, we study the maximization problem of a non-negative monotone $k$-submodular function under a knapsack constraint, and give a deterministic $(\frac12-\frac{1}{2e})$-approximation algorithm (see Theorem \ref{thm:main}). {It is an adaption to  Sviridenko's $(1-\frac1e)$-approximation algorithm for submodular knapsack maximization \cite{sviridenko2004note}.}

\medskip\noindent\textbf{Related works.}  It is well known that the \emph{diminishing return} property characterizes the submodular function. This property often appears in practice, and various problems can be formulated as submodular function maximization, {under different constraints}. Unfortunately, submodular function maximization is {generally} known to be NP-hard {\cite{feige1998threshold}}. Therefore, approximation algorithms that can run in polynomial time have been extensively studied. 

\noindent\emph{Submodular knapsack.}  {For monotone submodular maximization under a knapsack constraint, Sviridenko \cite{sviridenko2004note} presents a greedy $(1-\frac{1}{e})$-approximation algorithm with $O(n^5)$ queries, which enumerates all feasible sets of size no more than than 3 and then expands each set of size 3 greedily by the marginal density. This is the best possible approximation ratio among polynomial-time algorithms. Faster algorithms with $(1-\frac{1}{e}-\epsilon)$-approximation exist \cite{ene2019nearly}, but the time is exponential to $\frac{1}{\epsilon}$.  Yaroslavtsev \emph{et al.} \cite{yaroslavtsev2020bring} presented a Greedy+Max algorithm that is a $\frac12$-approximation with query complexity $O(\tilde Kn)$, where $\tilde K$ is an upper bound on the number of elements in any feasible solution. Huang \emph{et al.} \cite{huang2017streaming,huang2021improved} considered this problem in a streaming setting.  }

\noindent\emph{$k$-submodular maximization.} One decade ago, Huber and Kolmogorov \cite{huber2012towards} proposed $k$-submodular functions, which express the submodularity on choosing $k$ disjoint sets of elements instead of a single set, 
and recently become a popular subject of research \cite{gridchyn2013potts,hirai2016k,nguyen2020streaming,soma2019no}.
For unconstrained  non-monotone $k$-submodular maximization, Ward and {\v{Z}}ivn{\`y} \cite{ward2016maximizing}  proposed  a  $\max\{\frac13,\frac{1}{1 +a}\}$-approximation  algorithm with $a= \max\{1,\sqrt{\frac{k-1}{4}}\}$. Later, Iwata \emph{et al.} \cite{iwata2016improved} improved the approximation ratio to $\frac12$, which is improved to $\frac{k^2+1}{2k^2+1}$ by Oshima \cite{oshima2021improved} more recently. For unconstrained  \emph{monotone} $k$-submodular maximization, Ward and {\v{Z}}ivn{\`y} \cite{ward2016maximizing} {proved} that a greedy algorithm is $\frac12$-approximaion, and later, Iwata \emph{et al.} \cite{iwata2016improved} proposed a randomized $\frac{k}{2k-1}$-approximation algorithm, which is asymptotically tight.

In the constrained setting, Ohsaka and Yoshida \cite{ohsaka2015monotone} proposed a $\frac12$-approximation algorithm for nonnegative monotone $k$-submodular maximization with a total size constraint (i.e., $\cup_{i\in[k]}|X_i|\le B$ for an integer $B$) and a $\frac13$-approximation algorithm for that with individual size constraints (i.e., $|X_i|\le B_i~\forall i\in[k]$ with integers $B_i$). Sakaue \cite{sakaue2017maximizing} proposed a $\frac12$-approximation algorithm for nonnegative monotone $k$-submodular maximization with  a matroid  constraint {on the union of the sets.}
Thus, our work completes the picture by studying a knapsack constraint. 

\section{Preliminaries}


For notational ease, we identify $(k+1)^V$ with the set
{$$\mathscr S=\left\{\cup_{j=1}^t\{(a_j,i_j)\}~|~\forall t\in[|V|],a_j\in V~i_j\in[k]~\forall j\in[t], a_j\neq a_{j'}~\forall j\neq j'\right\}\cup \{\varnothing\}$$}
 that is, any $k$-disjoint set $\mathbf x=(X_1,\ldots,X_k)\in (k+1)^V$ uniquely corresponds to an item-index pairs set $S\in\mathscr S$, such that $(a_j,i_j)\in S$ if and only if $a_j\in X_{i_j}$. From now on, we rewrite $f(\mathbf x)$ as $f(S)$ with some abuse of notations, and thus $\Delta_{a,i}(S)$ means the marginal gain $f(S\cup\{(a,i)\})-f(S)$.
 For any $S\in\mathscr S$, we define $U(S):=\{a\in V~|~\exists i\in[k]~ \text{s.t.}~ (a,i)\in S\}$ to be the set of items included, and the \emph{size} of $S$ is $|S|=|U(S)|$.
 {In the remainder of this note, let $f$ be an arbitrary non-negative, monotone, $k$-submodular function. } We further assume that $f(\varnothing)=0$, which is without loss of generality because otherwise we can  redefine $f(S):=f(S)-f(\varnothing)$ for all $S\in\mathscr S$. 

We first introduce an important lemma.
\begin{lemma}\label{lem:21}
For any $S,S'\in\mathscr S$ with $S\subseteq S'$, we have
$$f(S')-f(S)\le \sum_{(a,i)\in S'\backslash S}\Delta_{a,i}(S).$$
\end{lemma}
\begin{proof}
Let $t=|S'|-|S|$. Define arbitrary manner sets $\{S_j\}_{j=0}^t$ such that (1) $S_0=S$; (2) $|S_{j}\backslash S_{j-1}|=1$ for $j\in[t]$; (3) $S_t=S'$. Let $\{(a_j,i_j)\}:=S_{j}\backslash S_{j-1}$ for $j\in[t]$. Then we have
\begin{align*}
    f(S')-f(S)&=\sum_{j=1}^t \Delta_{a_j,i_j}(S_{j-1})\le \sum_{j=1}^t \Delta_{a_j,i_j}(S),
\end{align*}
where the inequality follows from the orthant submodularity.
\end{proof}

The following proposition from Ward and {\v{Z}}ivn{\`y} \cite{ward2016maximizing} says that Greedy (see Algorithm \ref{alg:TS})  is $\frac{1}{2}$-approximation for maximizing $f$ without constraint (note that Theorem 5.1 of \cite{ward2016maximizing} states a more general conclusion which holds for a large class of $k$-set functions). Greedy considers items in an arbitrary order, and assign each item the \emph{best} index that brings the largest marginal gain.

\begin{algorithm}[t]
	\caption{\hspace{-2pt}~{\bf Greedy (without constraint)}
	}
	\label{alg:TS}
	\begin{algorithmic}[1]
	\REQUIRE  Set $V=\{1,2,\ldots,n\}$, monotone $k$-submodular function $f$
	\ENSURE A solution  $S\in\mathscr S$
	
\STATE $S\leftarrow \varnothing$
	\FOR{$a=1$ to $n$}
	\STATE  $i_a\leftarrow \arg\max_{i\in[k]}\Delta_{a,i}(S)$
\STATE $S\leftarrow S\cup\{(a,i_a)\}$
	\ENDFOR
	\RETURN $S$
	\end{algorithmic}
\end{algorithm}

\begin{proposition}[\cite{ward2016maximizing}]\label{thm:jie} 
Let $T\in \mathscr S$ be a solution that maximizes $f$ in the unconstrained setting, and $S\in \mathscr S$ be the solution returned by Greedy. Then $f(T)\le 2\cdot f(S)$.
\end{proposition}

In the later proofs we will use the following inequality from Wolsey \cite{wolsey1982maximising}.

\begin{lemma}[\cite{wolsey1982maximising}]\label{lem:wol}
{Let} $P$ and $D$ {be} arbitrary positive integers, {and} $(\rho_i)_{i=1}^P$ {be} arbitrary nonnegative {real values} with $\rho_1>0$. {Then}
$$\frac{\sum_{i=1}^P\rho_i}{\min_{t=1,\ldots,P}(\sum_{i=1}^{t-1}\rho_i+D\rho_t)}\ge 1-(1-\frac1D)^P\ge 1-e^{-P/D}.$$
\end{lemma}

\section{Knapsack constraint}

{Given a set} $V=\{1,\ldots,n\}$, nonnegative integers {$c_a\in\mathbb N$} for all $a\in V$, and budget $B\in \mathbb R$, we consider the following maximization problem with a  knapsack constraint,
\begin{equation}\label{pro}
\max\limits_{S\in\mathscr S}\left\{f(S): \sum_{a\in U(S)}c_a\le B\right\}.
\end{equation}

Sviridenko \cite{sviridenko2004note} considers the special case of $k=1$ (i.e., submodular maximization with a knapsack constraint), and presents a greedy $(1-\frac{1}{e})$-approximation algorithm {with $O(n^5)$ queries, which enumerates all feasible sets of size  at most 3 and then expands each set of size 3  greedily by the marginal density.}  We adapted it to Algorithm \ref{alg:line} for problem (\ref{pro}) {by enumerating all feasible sets of size  at most 2 and then expanding each set of size 2 greedily}, and prove an approximation ratio of $\frac12-\frac{1}{2e}$. For any solution $S\in\mathscr S$, define $c(S)$ to be the total cost of all items in $S$.

In Line 1 of Algorithm \ref{alg:line}, it enumerates all feasible {singleton} solutions, and store the currently best solution as $S_A$; it takes {$O(nk)$} oracle queries. Then it considers all feasible sets of size {two}, and completes each such set greedily with respect to the density, subject to the knapsack constraint. There are {$O(n^2k^2)$} such sets, and for each set it takes $O(n^2k)$ queries. Thus, the time complexity is {$O(n^4k^3)$}. 

\begin{algorithm}[H]
	\caption{\hspace{-2pt}~{\bf Greedy for (\ref{pro})}}
	\label{alg:line}
	\begin{algorithmic}[1]
	
	\STATE   Let {$S_A\in\arg\max\limits_{S:~|S|=1,c(S)\le B}f(S)$} be a {singleton} solution giving the largest value. 
	\FOR{ every  $I\in\mathscr S$ of {size 2}}
\STATE $S^0\leftarrow I$, $V^0\leftarrow V\backslash U(I)$
\FOR{$t$ from 1 to $n$}
	\STATE Let $\theta_t=\max\limits_{a\in V^{t-1},i\in[k]}\frac{\Delta_{a,i}(S^{t-1})}{c_a}$, and assume that the maximum is attained on $(a_t,i_t)$
	\IF{$c(S^{t-1})+c_{a_t}\le B$}
	\STATE $S^t=S^{t-1}\cup\{(a_t,i_t)\}$ 
\ELSE
\STATE $S^t=S^{t-1}$
	\ENDIF
	\STATE {$V^t=V^{t-1}\backslash\{a_t\}$}
\ENDFOR
\STATE $S_A\leftarrow S^n$ if $f(S^n)>f(S_A)$
	\ENDFOR
\RETURN $S_A$
	\end{algorithmic}
\end{algorithm}

\begin{theorem}\label{thm:main}
For maximizing $f$ under a knapsack constraint, Algorithm \ref{alg:line} has an approximation ratio of $\frac12-\frac{1}{2e}$, and evaluates $f$ {$O(n^4k^3)$} times.
\end{theorem}
\begin{proof}
 Let $T=\{(a_1^\ast,i_1^\ast),\ldots,(a_{|T|}^\ast,i_{|T|}^\ast)\}$ be an optimal solution. {If $|T|=1$}, our algorithm must find it in Line 1. So we only need to consider {$|T|\ge 2$}.   We order the set $T$ so that for any $t=1,\ldots,|T|$,
 $$f(T^t)=\max_{(a,i)\in T\backslash T^{t-1}}f(T^{t-1}\cup\{(a,i)\}),$$
 where {$T^t=\{(a_1^\ast,i_1^\ast),\ldots,(a_{t}^\ast,i_{t}^\ast)\}$}, and $T^0=\varnothing$.

 {Let $Y=T^2$ be the set that consists of the first two items of $T$. For any item $(a_j^\ast,i_j^\ast)\in T$, $j\ge3$, and any set $Z\subseteq V\backslash\{a_1^\ast,a_2^\ast,a_j^\ast\}\times [k]$, by the ordering of the sets in $T$,} we have
  \begin{align*}
   f(Y\cup Z\cup \{(a_j^\ast,i_j^\ast)\})-f(Y\cup Z)&\le f(T^1)-f(\varnothing)\le f(T^1);\\
   f(Y\cup Z\cup \{(a_j^\ast,i_j^\ast)\})-f(Y\cup Z)&\le  f(T^1\cup\{(a_j^\ast,i_j^\ast)\})-f(T^1)\le f(T^2)-f(T^1).
 \end{align*}
{It follows from the summation of the above two inequalities that 
 \begin{align}\label{eq:33}
 f(Y\cup Z\cup \{(a_j^\ast,i_j^\ast)\})-f(Y\cup Z)\le f(Y)/2.
 \end{align}
Now,  we consider the iteration in which the algorithm chooses set $Y$} at the beginning of the greedy procedure, i.e. $S^0 =Y$. Define a function $g(S)=f(S)-f(Y)$ for all $S\supseteq Y$, which is also a monotone $k$-submodular function.

 Let $\hat{t}+1$ be the first step that the algorithm does not add item $a_{\hat{t}+1}\in U(T)$ to the current set $U(S^{\hat{t}})$ because its addition would exceed the budget. Thus $S^{\hat{t}+1}=S^{\hat{t}}$. We can further assume that $\hat{t}+1$ is the first step $t$ for which $S^t=S^{t-1}$. This assumption is without loss of generality, because if it happens earlier for some $t'<\hat{t}+1$, then {$a_{t'}$ does not belong to the optimal solution $T$, nor} the approximate solution we are interested in; {thus, we can  remove $a_{t'}$ from the ground set $V$, without affecting} the analysis, the optimal solution $T$, and the approximate solution obtained in the iteration with $S^0 = Y$.

 Note that $Y\subseteq T\cap S^t$, {for any $t=0,\ldots,\hat{t}$}. 
 Define $OPT_g(V')$ to be the optimal value of function $g$ over items $V'\subseteq V$ without constraint. We greedily construct a set $\tilde S\in\mathscr S$ over items $U(T)\cup U(S^t)$: starting with $Y\subseteq \tilde S$, consider every item in $U(S^t\backslash Y)$ in the same order as it is added to $U(S^t)$ in Algorithm \ref{alg:line}, and then consider every item in $U(T)\backslash U(S^t)$ in an arbitrary order; when considering each item, assign the \emph{best} index that brings the largest marginal gain. Clearly $S^t\subseteq \tilde S$, as the indices in $S^t$ are assigned greedily.  For any $t=0,\ldots,\hat{t}$, we have
\begin{align}\label{eq:ss}
\begin{split}
 g(T)&\le OPT_g(U(T)\cup U(S^t))\le 2\cdot g(\tilde S)\\
 &\le 2\left(g(S^t)+\sum_{(a,i)\in \tilde S\backslash S^t}(g(S^t\cup\{(a,i)\})-g(S^t))\right)\\
 &= 2\left(g(S^t)+\sum_{(a,i)\in \tilde S\backslash S^t}(f(S^t\cup\{(a,i)\})-f(S^t))\right)\\
 &\le 2\left(g(S^t)+(B-c(Y))\theta_{t+1}\right).
\end{split}
\end{align}
 The second inequality follows from the fact that $\tilde S$ is obtained greedily and {thus achieves a 2-approximation by} Proposition \ref{thm:jie}. The third inequality is because of Lemma \ref{lem:21}.  The last inequality follows from $f(S^t\cup\{(a,i)\})-f(S^t)\le c_a\cdot \theta_{t+1}$ and $\sum_{(a,i)\in \tilde S\backslash S^t}c_a\le B-c(Y)$.

Let $B_t = \sum_{\tau = 1}^t c_{a_\tau}$ 
and $B_0 = 0$. Define $B^\prime = B_{\hat{t}+1}$ and $B^{\prime\prime} = B - c(Y)$. By the definition of the item $a_{\hat{t}+1}$, we have $B^\prime > B\ge B^{\prime\prime}$. For $j = 1,\dots, B^\prime$, we define $\rho_j = \theta_t$ if $j = B_{t-1}+1,\dots,B_t$ (that is, $\rho_1=\cdots=\rho_{B_1}=\theta_1$,  $\rho_{B_1+1}=\cdots=\rho_{B_2}=\theta_2$, $\ldots$,  $\rho_{B_{\hat{t}}+1}=\cdots=\rho_{B'}=\theta_{\hat{t}+1}$).
Using this definition, we obtain $g(S^t) = \sum_{\tau=1}^t c_{a_\tau}\theta_{\tau} = \sum_{j=1}^{B_t} \rho_j$ for $t=1,\dots,\hat{t}$ and $g(S^{\hat{t}}\cup\{(a_{\hat{t}+1},i_{\hat{t}+1})\}) = \sum_{\tau=1}^{\hat{t}+1}c_{a_\tau}\theta_{\tau} = \sum_{j=1}^{B^\prime} \rho_j$.
Then we have equalities
\begin{align}\label{eq:min_eq}
\begin{split}
\min_{s = 1,\dots,B^\prime}\{\sum_{j=1}^{s-1}\rho_j + B^{\prime\prime}\rho_s\}=~&\min_{t = 0,\dots,\hat{t}}\{\sum_{j=1}^{B_t}\rho_j + B^{\prime\prime}\rho_{B_t+1}\}\\
=~&\min_{t = 0,\dots,\hat{t}}\{g(S^t) + B^{\prime\prime}\theta_{t+1}\}.
\end{split}
\end{align}
Combining \eqref{eq:min_eq} with \eqref{eq:ss} and Lemma \ref{lem:wol}, we obtain
\begin{align}\label{eq:77}
\begin{split}
\frac{g(S^{\hat{t}}\cup\{(a_{\hat{t}+1},i_{\hat{t}+1})\})}{g(T)} =~& \frac{\sum_{j=1}^{B^\prime}\rho_j}{g(T)}\\
\geq~&\frac{\sum_{j=1}^{B^\prime}\rho_j}{2\cdot\min_{s = 1,\dots,B^\prime}\{\sum_{j=1}^{s-1}\rho_j + B^{\prime\prime}\rho_s\}}\\
\geq~&\frac{1}{2}(1-e^{-B^\prime/B^{\prime\prime}})>\frac{1}{2}(1-e^{-1}).
\end{split}
\end{align}

Finally, combining \eqref{eq:33}  and \eqref{eq:77}, we obtain {a lower bound on the output $f(S_A)$ of our algorithm:}
\begin{align}
f(S_A)\ge f(S^{\hat{t}}) =~& f(Y) + g(S^{\hat{t}})\nonumber\\
=~&f(Y) + g(S^{\hat{t}}\cup\{(a_{\hat{t}+1},i_{\hat{t}+1})\})-g(S^{\hat{t}}\cup\{(a_{\hat{t}+1},i_{\hat{t}+1})\})+g(S^{\hat{t}})\nonumber\\
=~&f(Y) + g(S^{\hat{t}}\cup\{(a_{\hat{t}+1},i_{\hat{t}+1})\})-(f(S^{\hat{t}}\cup\{(a_{\hat{t}+1},i_{\hat{t}+1})\})-f(S^{\hat{t}}))\nonumber\\
\geq~&f(Y) + \frac{1}{2}(1-e^{-1})g(T)-f(Y)/2\label{eq:last}\\
\geq~&\frac{1}{2}(1-e^{-1})f(T).\nonumber
\end{align}
\end{proof}

\section{{Discussions}}
{We remark that the  {proof idea} of Theorem 1 generally follows the proof of the $(1-\frac1e)$-approximation algorithm for submodular knapsack maximization in \cite{sviridenko2004note}. That is, first upper bound the marginal value brought by an single item in the optimal solution on the basis of $Y$ (as in (\ref{eq:33})), then upper bound the optimal value $g(T)$ (as in (\ref{eq:77})), and finally derive a lower bound on our solution $f(S^A)$. 
There are two {main} differences. First, we reduce the enumeration in the algorithm from subsets of size three in \cite{sviridenko2004note} to two. The reason we can do it is that, if the starting set $Y$ is of size $s$, then the RHS of Eq. (\ref{eq:33}) is $\frac{f(Y)}{s}$; this will be used to derive Eq. (\ref{eq:last}), which only requires $1-\frac1s\ge \frac12(1-e^{-1})$. Thus, a size $s=2$ is enough for the analysis. Second, the  additional difficulty that arises in our problem is that, we can no longer obtain an upper bound on $g(T)$ straightforwardly from Lemma \ref{lem:21}, because the items in $U(T)\cap U(S^t)$ may have different indices in $T$ and $S^t$.  To overcome this, we construct an intermediary $\tilde S$ over the items in $U(T)\cap U(S^t)$ by the Greedy algorithm, such that $g(T)\le 2\cdot g(\tilde S)$  (see Eq. (\ref{eq:ss})), and then use  Lemma \ref{lem:21} to upper bound the value $g(\tilde S)$. Thus, we indirectly obtain an upper bound on $g(T)$. }

{One may be surprised by the fact that we get a $\frac12(1-\frac1e)$-approximation for any $k$, whereas there is a $(1-\frac1e)$-approximation when $k=1$ \cite{sviridenko2004note}. 
The reason for such a jump is that, the approximation ratio of the Greedy (which is used to construct $\tilde S$) is trivially 1 for $k=1$, but jumps to 2 for any $k\ge 2$ (and the analysis is tight). It would be an interesting direction to get an approximation ratio that degrades smoothly as a function of $k$. }

\bibliography{reference}

\end{document}